\documentclass[10pt,letterpaper]{article}
\usepackage[top=0.85in,left=2.75in,footskip=0.75in]{geometry}

\usepackage{amsmath,amssymb}

\usepackage{changepage}

\usepackage[utf8x]{inputenc}

\usepackage{textcomp,marvosym}

\usepackage{cite}

\usepackage{nameref,hyperref}

\usepackage[right]{lineno}

\usepackage{microtype}
\DisableLigatures[f]{encoding = *, family = * }

\usepackage[table]{xcolor}

\usepackage{array}

\newcolumntype{+}{!{\vrule width 2pt}}

\newlength\savedwidth

\raggedright
\usepackage[aboveskip=1pt,labelfont=bf,labelsep=period,justification=raggedright,singlelinecheck=off]{caption}

\makeatletter
\renewcommand{\@biblabel}[1]{\quad#1.}
\makeatother

\usepackage{lastpage,fancyhdr,graphicx}
\usepackage{epstopdf}
\fancyheadoffset[L]{2.25in}
\fancyfootoffset[L]{2.25in}
\newcommand{\fabio}{\textcolor{red}}

\newcommand{\problemname}{Target Associated $k$-Set}
\newcommand{\algname}{UNCOVER}
\newcommand{\E}{\mathbf{E}}
\newcommand{\BOM}[1]{{\Omega}\left(#1\right)}

\usepackage{paralist}
\usepackage{amsthm}
\newtheorem{theorem}{Theorem}
\newtheorem{proposition}{Proposition}
\usepackage{algorithm2e}
\usepackage{url}

\begin{document}
\vspace*{0.2in}

% Title must be 250 characters or less.
\begin{flushleft}
{\Large
\textbf\newline{Efficient algorithms to discover alterations with complementary functional association in cancer} % Please use "sentence case" for title and headings (capitalize only the first word in a title (or heading), the first word in a subtitle (or subheading), and any proper nouns).
}
\newline
% Insert author names, affiliations and corresponding author email (do not include titles, positions, or degrees).
\\
Rebecca Sarto Basso\textsuperscript{1},
Dorit S. Hochbaum\textsuperscript{1},
Fabio Vandin\textsuperscript{2,3,4*}
\\
\bigskip
\textbf{1} Department of Industrial Engineering and Operations Research, University of California at Berkeley, CA, USA.
\\
\textbf{2} Department of Information Engineering, University of Padova, Padova, Italy.
\\
\textbf{3} Department of Computer Science, Brown University, Providence, RI, USA.
\\
\textbf{4} Department of Mathematics and Computer Science, University of Southern Denmark, Odense, Denmark.
\\
\bigskip

% Insert additional author notes using the symbols described below. Insert symbol callouts after author names as necessary.
% 
% Remove or comment out the author notes below if they aren't used.
%
% Primary Equal Contribution Note
%\Yinyang These authors contributed equally to this work.

% Additional Equal Contribution Note
% Also use this double-dagger symbol for special authorship notes, such as senior authorship.
%\ddag These authors also contributed equally to this work.

% Current address notes
%\textcurrency Current Address: Dept/Program/Center, Institution Name, City, State, Country % change symbol to "\textcurrency a" if more than one current address note
% \textcurrency b Insert second current address 
% \textcurrency c Insert third current address

% Deceased author note
%\dag Deceased

% Group/Consortium Author Note
%\textpilcrow Membership list can be found in the Acknowledgments section.

% Use the asterisk to denote corresponding authorship and provide email address in note below.
% * fabio.vandin@unipd.it
* Corresponding author: fabio.vandin@unipd.it

\end{flushleft}
% Please keep the abstract below 300 words
\section*{Abstract}

Recent large cancer studies have measured somatic alterations in an unprecedented number of tumours. These large datasets allow the identification of cancer-related sets of genetic alterations by identifying relevant combinatorial patterns. Among such patterns, \emph{mutual exclusivity} has been employed by several recent methods that have shown its effectivenes in characterizing  gene sets associated to cancer.  Mutual exclusivity arises because of the complementarity, at the functional level, of alterations in genes which are part of a group (e.g., a \emph{pathway}) performing a given function. The availability of quantitative target profiles, from genetic perturbations or from clinical phenotypes, provides additional information that can be leveraged to improve the identification of cancer related gene sets by discovering groups with complementary functional associations with such targets.

In this work we study the problem of finding groups of mutually exclusive alterations associated with a quantitative (functional) target. We propose a combinatorial formulation for the problem, and prove that the associated computation problem is computationally hard. We design two algorithms to solve the problem and implement them in our tool \algname. We provide analytic evidence of the effectiveness of \algname\ in finding high-quality solutions and show experimentally that \algname\ finds sets of alterations significantly associated with functional targets in a variety of scenarios. In particular, we show that our algorithms find sets which are better than the ones obtained by the state-of-the-art method, even when sets are evaluated using the statistical score employed by the latter. In addition, our algorithms are much faster than the state-of-the-art, allowing the analysis of large datasets of thousands of target profiles from cancer cell lines. We show that on one such dataset from project Achilles our methods identify several significant gene sets with complementary functional associations with targets. Software available at: \url{https://github.com/VandinLab/UNCOVER}.

% Use "Eq" instead of "Equation" for equation citations.
\section*{Introduction}
Recent advances in sequencing technologies now allow to collect genome-wide measurements in large cohorts of cancer patients (e.g.,~\cite{Brennan:2013aa,Cancer-Genome-Atlas-Network:2015aa,Cancer-Genome-Atlas-Research-Network:2013aa,Cancer-Genome-Atlas-Research-Network:2014aa,cancer2017integrated,cancer2017integratedB}). In particular, they allow the measurement of the entire complement of somatic (i.e., appearing during the lifetime of an individual) alterations in all samples from large tumour cohorts. The study of such alterations has lead to an unprecedented  improvement in our understanding of how tumours arise and progress~\cite{garraway2013lessons}. One of the main remaining challenges is the interpretation of such alterations, in particular identifying alterations with functional impact or with relevance to therapy~\cite{mcgranahan2017clonal}.

Several computational and statistical methods have been recently designed to identify \emph{driver} alterations, associated to the disease, and to distinguish them from random, \emph{passenger} alterations not related with the disease~\cite{pmid24479672}. The identification of genes associated with cancer is complicated by the extensive \emph{intertumour heterogeneity}~\cite{vandin2017computational}, with large (100-1000's) and different collections of alterations being present in tumours from different patients and no two tumours having the same collection of alterations~\cite{vogelstein2013cancer,vandin2017computational}. Two main reasons for such heterogeneity are that \begin{inparaenum}[i)]
\item most mutations are passenger, \emph{random} mutations, and,  more importantly, 
\item driver alterations target cancer \emph{pathways}, 
\end{inparaenum}
groups of interacting genes that perform given functions in the cell and whose alteration is required to develop the disease. Several methods have been designed to identify cancer genes using \textit{a-priori} defined pathways~\cite{vaske2010inference} or interaction information in the form of large interaction networks~\cite{Leiserson:2015aa,creixell2015pathway}.

Recently several methods (see Section~\ref{sec:relwork}) for the \emph{de novo} discovery of mutated cancer pathways have leveraged the \emph{mutual exclusivity} of alterations in cancer pathways. Mutual exclusivity of alterations, with sets of genes displaying at most one alteration for each patient, has been observed in various cancer types~\cite{kandoth2013mutational,vogelstein2013cancer,garraway2013lessons,hanahan2011hallmarks}.  The mutual exclusivity property is due to the complementarity of genes in the same pathway, with alterations in different members of a pathway resulting in a similar impact at the functional level. Mutual exclusivity has been successfully used to identify cancer pathways in large cancer cohorts~\cite{kandoth2013mutational,Ciriello:2012ly,leiserson2015comet}.

An additional source of information that can be used to identify genes with complementary functions are quantitative measures for each samples such as: functional profiles, obtained for example by genomic or chemical perturbations~\cite{pmid25984343,pmid27260156,pmid28753430};  clinical data describing, obtained for example by (quantitative) indicators of response to therapy; activation measurements for genes or sets of genes, as obtained for example by single sample scores of Gene Set Enrichment Analysis~\cite{pmid12808457,pmid16199517}. The employment of such quantitative measurements is crucial to identify meaningful complementary alterations since one can expect mutual exclusivity to reflect in functional properties (of altered samples) that are specific to the altered samples.

\subsection*{Related work}
\label{sec:relwork}

Several recent methods have used mutual exclusivity signals to identify sets of genes important for cancer~\cite{yeang2008combinatorial}. RME~\cite{Miller:2011ve} identifies mutually exclusive sets using a score derived from information theory.  Dendrix~\cite{Vandin:2012ys} defines a combinatorial gene set score and uses a Markov Chain Monte Carlo (MCMC) approach for identifying mutually exclusive gene sets altered in a large fraction of the patients. Multi-Dendrix~\cite{leiserson2013simultaneous} extends the score of Dendrix to multiple sets and uses an integer linear program (ILP) based algorithm to simultaneously find multiple sets with mutually exclusive alterations. CoMET~\cite{leiserson2015comet} uses a generalization of Fisher exact test to higher dimensional contingency tables to define a score to characterize mutually exclusive gene sets altered in relatively low fractions of the samples. WExT~\cite{leiserson2015comet} generalizes the test from CoMET to incorporate individual gene  weights (probabilities) for each alteration in each sample. WeSME~\cite{kim2016wesme} introduces a test that incorporates the alteration rates of patients and genes and uses a fast permutation approach to assess the statistical significance of the sets. TiMEx~\cite{constantinescu2015timex} assumes a generative model for alterations and defines a test to assess the null hypothesis that mutual exclusivity of a gene set is due to the interplay between waiting times to alterations and the time at which the tumor is sequenced. MEMo~\cite{Ciriello:2012ly} and the method from~\cite{babur2015systematic} employ mutual exclusivity to find gene sets, but use an interaction network to limit the candidate gene sets. The method by \cite{raphael2015simultaneous} and PathTIMEx~\cite{cristea2016pathtimex} introduce an additional dimension to the characterization of inter-tumor heterogeneity, by reconstructing the order in which mutually exclusive gene sets are mutated. None of these methods take quantitative targets into account in the discovery of significant gene sets and sets showing high mutual exclusivity may not be associated with target profiles (Fig.~\ref{fig:motivation}).

\begin{figure}
\begin{center}
\includegraphics[width=0.8\textwidth]{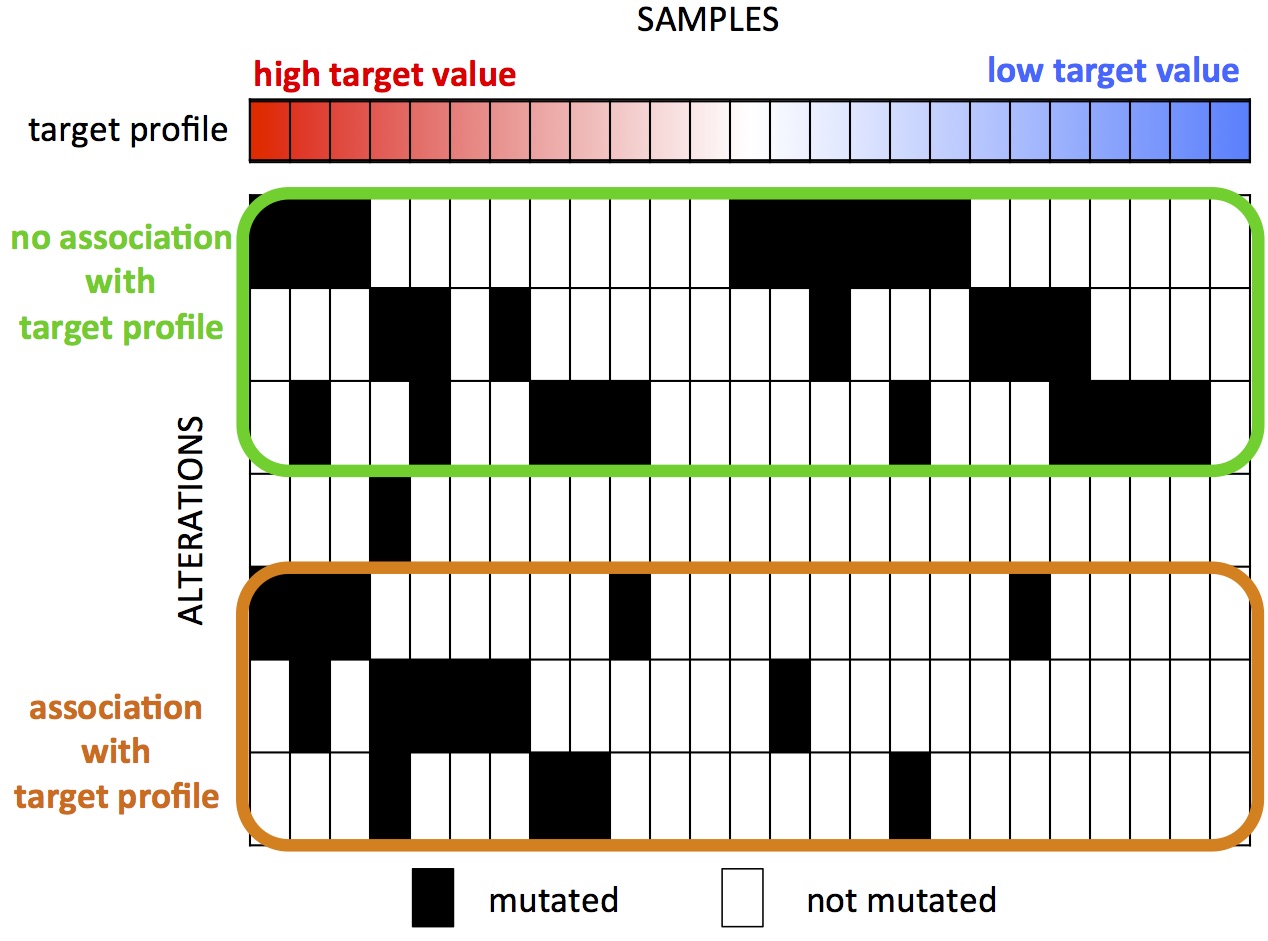}
\caption{\textbf{Identification of mutually exclusive alterations associated with a target profile.}\\ Alterations in the green set have high mutual exclusivity but no association with the target profile. Alterations in the orange set have lower mutual exclusivity but strong association with the target profile. Methods that find mutually exclusive sets of alterations without considering the target profile will identify the green set as the most important gene set.}
\label{fig:motivation}
\end{center}
 \end{figure}

\cite{revealer} recently developed the repeated evaluation of variables conditional entropy and redundancy (REVEALER) method, to identify mutually exclusive sets of alterations associated with functional phenotypes. REVEALER uses as objective function (to score a set of alterations) a re-scaled mutual information metric called \emph{information coefficient} (IC).  REVEALER employs a greedy strategy, computing at each iteration the conditional mutual information (CIC) of the target profile and each feature, conditioned on the current solution. REVEALER can be used to find sets of mutually exclusive alterations starting either from a user-defined seed for the solution or from scratch, and~\cite{revealer} shows that REVEALER finds sets of meaningful cancer-related alterations.

\subsection*{Our contribution}

In this paper we study the problem of finding sets of alterations with complementary functional associations using alteration data and a quantitative (functional) target measure from a collection of cancer samples. Our contributions in this regard are fourfold. First,
we provide a rigorous combinatorial formulation for the problem of finding groups of mutually exclusive alterations associated with a quantitative target and prove that the associated computational problem is NP-hard. Second, we develop two efficient algorithms, a greedy algorithm and an ILP-based algorithm to identify the set of $k$ genes with the highest association with a target; our algorithms are implemented in our method fUNctional Complementarity of alteratiOns discoVERy (\algname). Third, we show that our algorithms identify highly significant sets of genes in various scenarios; in particular, we compare \algname\ with REVEALER on the same datasets
used in~\cite{revealer}, showing that \algname\ identifies solutions of  higher quality than REVEALER while being on average two order of magnitudes faster than REVEALER. Interestingly, the solutions obtained by \algname\ are better than the ones obtained by REVEALER even when evaluated using the objective function (IC score) optimized by REVEALER.
Fourth, we show that the efficieny of \algname\ enables the analysis of a large dataset from Project Achilles with thousands of target measurements and tens of thousands of alterations. On such dataset \algname\ identifies identifies several statistically significant associations between target values and mutually exclusive alterations in genes sets
sets,  with an enrichment in well-known cancer genes and in known cancer pathways.

\section*{Materials and methods}

This section describes the problem we study and the algorithms we designed to solve it, that are implemented in our tool \algname. We also describe the data and computational environment for our experimental evaluation.

\subsection*{\algname: Functional complementarity of alterations discovery}

The workflow of our algorithm \algname\ is presented in Fig.~\ref{fig:workflow}. \algname takes in input information regarding \begin{inparaenum}
\item the alterations measured in a number of samples (e.g., patients or cell lines), and
\item the value of the \emph{target} measure for each patient.
\end{inparaenum}
\algname\ then identifies the set of mutually exclusive alterations with the highest association to the target, and employs a permutation test to assess the significance of the association. Details regarding the computational problem and the algorithms used by \algname\ are described in the following sections. The implementation of \algname\ is available at \url{https://github.com/VandinLab/UNCOVER}.

\begin{figure}
\begin{center}
\includegraphics[width=\textwidth]{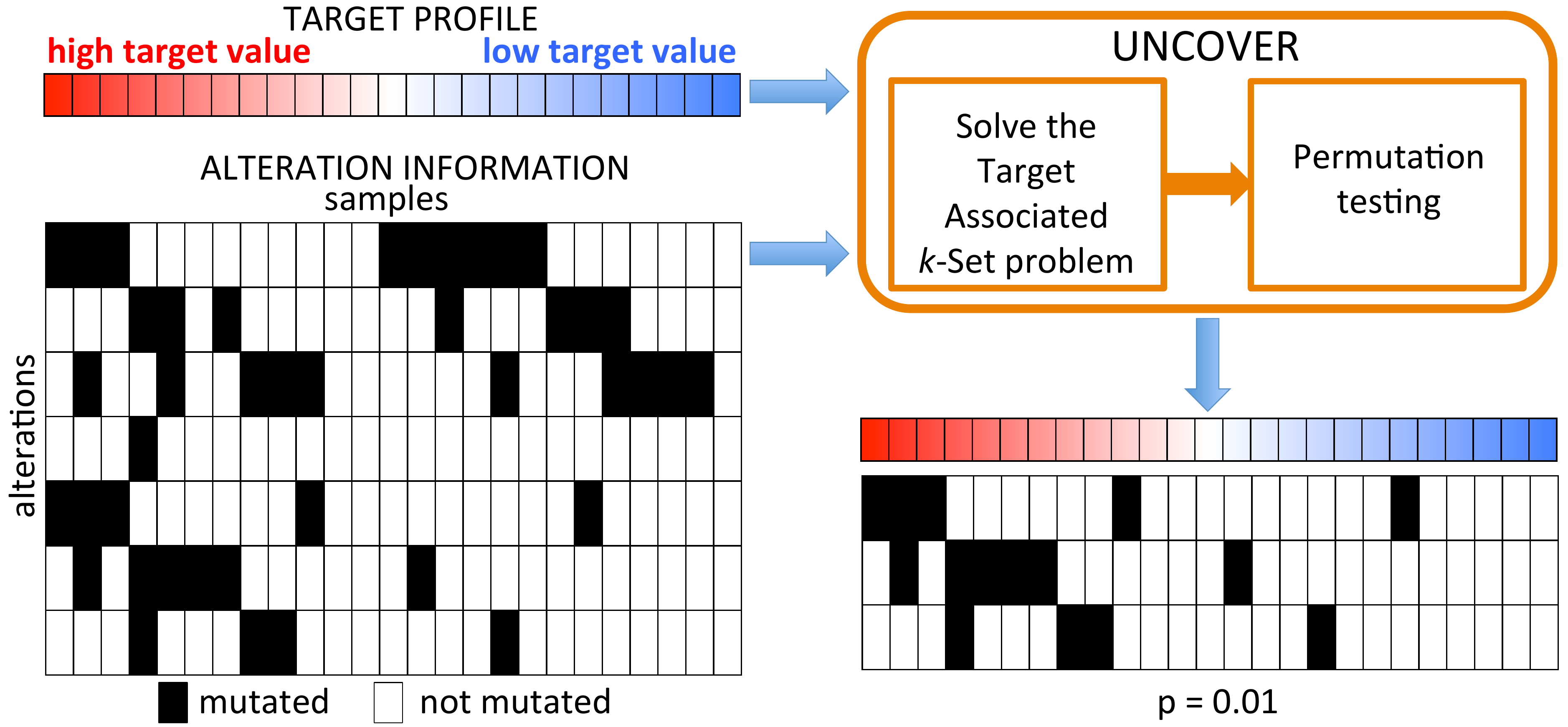}
\caption{\textbf{\algname: Functional complementarity of alterations discovery.}\\\algname\ takes in input the alterations information and a target profile for a set of samples, and identifies the set of complementary alterations with the highest association to the target by solving the \problemname\ problem and performing a permutation test.}
\label{fig:workflow}
\end{center}
 \end{figure}
 
\subsubsection*{Computational problem}
\label{sec:problem}

Let $J=\{j_1,\dots,j_m\}$ be the set of samples and let $G = \{g_1,\dots, g_n\}$ be the set of genes for which we have measured alterations in $J$. We are also given a \emph{target profile}, that is for each sample $j \in J$ we have a target value $w_j \in \mathbb{R}$ which quantitatively measures a functional phenotype (e.g., pathway activation, drug response, etc.). For each sample $j\in J$ we also have information on whether each $g \in G$ is altered or not in $j$. 
Let $A_g$ be the set of patients in which gene $g \in G$ is mutated. We say that a patient $j \in J$ is covered by gene $g\in G$ if $j \in A_g$ i.e. if gene $g$ is mutated in sample $j$. Given a set of genes $S \subset G$, we say that sample $j \in J$ is covered by $S$ if $j \in \cup_{g \in S} A_g$.

The goal is to identify a set $S$ of at most $k$ genes, corresponding to $k$ subsets $S_1, S_2,\dots S_k$ where for each subset $S_i$ we have that $S_i \subseteq J$, such that the sum of the weights of the elements covered by $S$ is maximized. We also penalize overlaps between sets when an element is covered more than once by $S$ by assigning a penalty $p_j$ for each of the additional times $j$ is covered by $S$. As penalty we use the positive average of the normalized target values if the original weight of the element was positive. If the original weight of the element was negative we assign a penalty equal to its weight.

Let $c_S(j)$ be the number of sets in $S_1, \dots, S_k$ that cover element $j \in J$.

Therefore for a set $S$ of genes, we define its weight $W(S)$ as:

\begin{equation*}
    W(S) = \sum_{j \in \cup_{s \in S}{A_s}} w_j - \sum_{j \in \cup_{s \in S}{A_s}}(c_S(j)-1) p_j
\end{equation*}

\begin{problem}{The \problemname\ problem}
\fabio{Given a set $J$ of samples, sets $A_{g_1},\dots,A_{g_n}$ describing alterations of genes $G=\{g_1,\dots,g_n\}$ in the set $J$, weights $w_j$ and penalties $p_j>0$ for each sample $j \in J$ find the $S$ of $\le k$ elements maximizing $W(S)$.}
\end{problem}

The following results defines the computational hardness of the problem above.

\begin{theorem}
\label{thm:nphard}
The \problemname\  problem is NP-hard.
\end{theorem}
\begin{proof}[Proof]
The proof is by reduction from the Maximum Weight Submatrix Problem (MWSP) defined and proved to be NP-hard in~\cite{pmid21653252}. The MSWP takes as input an $m\times  n$ binary matrix $A$ and an integer $k > 0$ and requires to find the $m \times k$ column sub-matrix $\hat{M}$ of $A$ that maximizes the objective function $|\Gamma(M)| - \omega(M)$, where $\Gamma(M)$ is the set of rows with at least one $1$ in columns of $M$ and $\omega(M)=\sum_{g \in M} |\Gamma(\{g\})| - |\Gamma(M)|$.

Given an instance of Maximum Weight Submatrix Problem, we define an instance of the \problemname\ as follow: the set of samples $J$ corresponds to the rows of $A$, the set of genes $G$ corresponds to the columns of $A$, and the set $S_g$ of samples covered by gene $g\in G$ is the subset of the rows in which $g$ has a 1. By setting $w_j = 1$ and $p_j = 1$ for all $j \in J$, we have that the objective function of MWSP corresponds to the weight $W(S)$ for the \problemname\, therefore the optimal solution of the \problemname\ corresponds to the optimal solution of MWSP.
\end{proof}

\subsubsection*{ILP formulation}
\label{sec:ILP}

In this subsection we provide an ILP formulation for the \problemname\ problem. Let $x_i$ be a binary variable equal to $1$ if set $i \in G$ is selected and $x_i=0$ otherwise. Let $z_j$ be a binary variable equal to $1$ if element $j$ is covered and $z_j=0$ otherwise. Let  $y_j$ denote the number of sets in the  solution covering element $j$. Finally, let $w_j$ be the weight of element $j$ and  $p_j$ be the penalty for element $j$

In our ILP formulation, the following constraints need to be satisfied by a valid solution:
\begin{itemize}
    \item the total number of sets in the solution is at most $k$: $\sum_{i} x_{i} \le k$
    \item for each element $j \in J$ we have:  $y_j=\sum_{i: j \in S_i} x_{i}$
    \item for each element $j \in J$, if $j$ is covered by the current solution then the number of times $j$ is covered in the solution is at least $1$:  $y_{j} \ge z_{j}$
    \item for each element $j \in J$, if $j$ is covered by at least one element in the current solution then $j$ is covered: $z_{j} \ge y_{j}/k$.
\end{itemize}

With the variables defined above, the score for a given solution is 
\begin{equation}
z(q) = \max \sum_{j=1}^m (w_{j}+p_j) z_{j}-\sum_{j=1}^m p_j y_{j}.
\end{equation}

$z(q)$ constitutes the objective function of our ILP formulation.

\subsubsection*{Greedy algorithm}
\label{sec:alg}

Since solving ILPs can be impractical for very large datasets, we also design a $k$-stage greedy algorithm to solve the \problemname\ problem. During each stage the algorithm picks 1 set $A_i$ to be part of the solution; this is done by first computing the total weight of each subset which is defined as the sum of the weights of its elements $W(A_i)=\sum_{j \in A_i} w_{j}$. Then the algorithm finds the subset $A_i$ of maximum positive weight and adds it to the solution. It may be that at some stage $\ell$ no additional set of positive weight can be selected, in this case, the solution obtained after stage $\ell-1$ will be our output.  
At the end of the iteration the weight of every element $j$ that belonged to the chosen set $A_i$ is set to the negative of penalty $p_j$, in order to penalize future selections of the same elements. The greedy algorithm is described in Algorithm~\ref{algo:change}.

\begin{algorithm}
\SetAlgoLined
%\DontPrintSemicolon % Some LaTeX compilers require you to use \dontprintsemicolon instead 
\KwIn{A set of elements $J$ (samples), a class $I$ of subsets of $J$ (genetic alterations) and an integer $k$ (number of alterations we want to find). Each element $j \in J$  has an associated weight $w_j$ (target profile) and a penalty $p_j$.}
\KwOut{$k$  subsets, $S_1, S_2,\dots S_k$ , where each subset selected is a member of $I$, such that the sum of the weights of the elements in the selected sets is maximized and the overlap between selected sets is minimized.}
\For{$\ell\gets 1$ to $k$}{
    \lFor{$i \gets 1$ to $n$}{$W(A_i) \gets \sum_{j \in A_i} w_{j}$}
    $S_{\ell} \gets \arg\max_{A_i > 0}\{W(A_i)\}$\;
    \lFor{$j \in S_{\ell}$}{$w_j \gets -p_j$}
}
\textbf{return} $S_1 \dots S_k $\; 
\caption{{\sc Greedy} Coverage}
\label{algo:change}
\end{algorithm}

We note that our greedy algorithm is analogous to the greedy algorithm for the Maximum k-Coverage problem~\cite{kCoverage} with the difference that rather than eliminating the elements already selected we change their weight to a penalty. Also, assuming it is acceptable to return less than $k$ sets, we only pick a set if it has a positive weight. The running time of the algorithm is $O(kmn)$ where $m =$ number of samples and $n =$ number of alterations.

While the greedy algorithm may not return the optimal solution,
we prove that it provides guarantees on the weight of the solution it provides.

\begin{proposition}
Let $S^*$ the optimal solution of the \problemname\ and $\hat{S}$ be the solution returned by the greedy algorithm. Then $W(\hat{S}) \ge W(S^*)/k$.
\end{proposition}
\begin{proof}

Note that the weight of subsets in the optimal solution $W(S^*)$ can only be lower compared to the original weight of the subsets, since the only weight update operation performed is to substitute positive weights of elements selected with a negative penalty. 

The first subset $\hat{S_1}$  selected by our algorithm is the set of maximum weight out of all subsets and therefore $W(\hat{S_1}) \ge W(S_{\ell}^*)$ for $\ell=1..k$. By the pigeonhole principle, one of these subsets in the optimal solution must cover at least $W(S^*)/k$ worth of elements. Thus $W(\hat{S_1}) \ge W(S^*)/k$.
Therefore the first subset selected by the algorithm already gives a $1/k$ approximation of the optimal solution. In subsequent iterations of the algorithm we only pick additional sets if they have a positive weight so our approximation ratio can only improve. 
\end{proof}

We also prove that the bound above is tight

\begin{proposition}
There are instances of the \problemname\ such that $W(\hat{S}) = W(S^*)/k$.
\end{proposition}

The proof is in the Supplementary Material.

While the proposition above is based on an extreme example, our experimental analysis shows that in practice the greedy algorithm works well and often identifies the optimal solution. We therefore analyze the greedy algorithm under a generative model in which there is a set $H$ of $k$ genes with mutually exclusive alterations associated with the target, while each genes $g  \in G \setminus H$ is mutated in sample $j$ with probability $p_g$ independently of all other events. We also assume that the weights $w_j$ are such that $\sum_{j\in J} w_j = 0$ and for each $j: |w_j \le 1|$. (In practice this is achieved by normalizing the target values before running the algorithm, by subtracting to each $w_j$ the average value $\sum_{j\in J} w_j/m$ and then dividing the result by the maximum of the absolute values of the transformed $w_j$'s.) Note that this last condition implies that $|p_j| \le 1$ for all $j$. We also assume that for genes in $H$ the following assumptions hold:
\begin{itemize}
    \item the set $H$ has an association with the target, i.e.: $\E[W(H)] \ge \frac{m}{c'}$ for a constant $c'\ge 1$.
    \item each gene of $H$ contributes to the weight of $H$, i.e. for each $S \subset H$ and each $g \in H \setminus S$ we have $\E[W(S \cup\{g\})] - \E[W(S)]\ge \frac{W(H)}{k c^{''}}$ for a constant $c^{''}\ge 1$.
\end{itemize}

The following shows that, if enough samples from the generative model are considered, the greedy algorithm finds the set $H$ associated with the target with high probability.

\begin{proposition}
If $m \in \BOM{k^2 \ln (n/\delta)}$ samples from the generative model above are provided to the greedy algorithm, then the solution of the greedy algorithm is H with probability $\ge \delta$.
\end{proposition}

The proof is in the Supplementary Material.

\subsubsection*{Statistical significance}
\label{sec:perm_test}

To assess the significance of the solution reported by our algorithms we use a permutation test in which the dependencies among alterations in various genes are maintained, while the association of alterations and the target is removed. In particular, a permuted dataset under the null distribution is obtained as follows: the sets $A_g, g\in \mathcal{G}$ are the same as observed in the data; the values of the target are randomly permuted across the samples. 

To  estimate the $p$-value for the solutions obtained by our methods we used the following standard procedure: \begin{inparaenum}[1)] \item we run an algorithm (ILP or greedy) on the real data $\mathcal{D}$, obtaining a solution with objective function $o_{\mathcal{D}}$; \item we generate $N$ permuted datasets as described above; \item we run the same algorithm on each permuted dataset; \item the $p$-value is the given by $(e+1)/(N+1)$, where $e$ is the number of permuted datasets in which our algorithm found a solution with objective function $\ge o_{\mathcal{D}}$.\end{inparaenum}

\subsection*{Data and computational environment}
\label{sec:data}

\paragraph{Alteration Data.} We downloaded data for the Cancer Cell Line Encyclopedia on $25^{th}$ September, 2017 from 
\url{http://www.broadinstitute.org/ccle}. In particular we used the mutation (single nucleotide variants) and copy number aberrations (CNAs) which are derived from the original Cancer Cell Line Encyclopedia (CCLE) mutations and CNA datasets. The file we used is \texttt{CCLE\_MUT\_CNA\_AMP\_DEL\_0.70\_2fold.MC.gct}. It consists of a binary (0/1) matrix across 1,030 samples and 48,270 gene alterations in the form of mutations, amplifications and deletions, with a 1 meaning that the alteration is present in a sample, and a 0 otherwise.

\paragraph{Target Data.} In terms of target values we use the same datasets used by~\cite{revealer} to compare the performance of \algname\ with REVEALER.
In particular we used the following files available through the Supplementary Material of~\cite{revealer}: \texttt{CTNBB1\_transcriptional \_reporter.gct}, which consists of measurements of a $\beta$-catenin reporter in 81 cell lines; \texttt{NFE2L2\_activation\_profile.gct}, which includes NFE2L2  enrichment profiles for 182 lung cell lines; \texttt{MEK\_inhibitor\_profile.gct}, which contains MEK-inhibitor PD-0325901 sensitivity profile in 493 cancer cell lines from the Broad Novartis CCLE14l; and \texttt{KRAS\_essentiality\_profile.gct}, which  corresponds to the feature KRAS from a subset of 100 cell lines from the Achilles project dataset. 
In all these cases we considered the same direction of association (positive or negative) between alterations and the target as in~\cite{revealer}. Since our algorithm is very efficient we then decided to run it on a large dataset from Project Achilles~ (\url{https://portals.broadinstitute.org/achilles}), that uses genome-scale RNAi and CRISPR-Cas9 perturbations to silence or knockout individual genes. In particular, we use the whole 2.4.2 Achilles dataset (\texttt{Achilles\_QC\_v2.4.3.rnai.Gs.gct}) available from the project website. This dataset provides phenotype values for 5711 targets, corresponding to genes silenced by shRNA. The phenotype values correspond to ATARiS~\cite{pmid23269662} gene (target) level scores, quantifying the cell viability when the target gene is silenced by shRNA. These scores are provided for 216 cell lines~\cite{pmid25984343}, with 205 of them appearing in CCLE.

\paragraph{Data Preprocessing.} To be consistent with REVEALER we discarded features with high or low frequency, in particular features present in less than 3 samples or more than 50 samples were excluded from our analyses. The only exception was the MEK-inhibitor example, where the high frequency threshold was changed to be 100 since the number of original samples was substantially higher (i.e., 493) for this case.  From the Achilles dataset we excluded targets that have at least one missing value, in particular in this case we exclude 21 of the 5711 sets of target scores.
In all our experiments we normalized the target values before running the algorithm, by subtracting to each weight $w_j$ the average value $\sum_{j\in J} w_j/m$ and dividing the result by the standard deviation of the (original) $w_j$'s, in order to have both positive and negative target values.

\paragraph{Simulated Data.} 
We investigated how effective  \algname\ is at finding selected alterations in a controlled setting, where the ground truth is known. We generated target values according to a normal distribution with mean 0 and standard deviation 1. We tested dataset with 200, 600, 1000 and 10000 samples. For each dataset we considered the 38002 gene alterations present in CCLE and for each of them we placed alterations in the samples independently of all other events with the same frequencies as they appear in CCLE. To be consistent with the preprocessing done on rel data we filtered alterations to only have alterations with frequencies between 0.1 and 0.25. We also generated a set $T$ of 5 features to have an association with the target values. This association was varied throughout the experiments to cover different percentages of positive and negative targets. In particular we generated the selected features to cover 100\%, 80\%, 60\%, 40\% of the positive target values and 5\%, 10\%, 15\%, 20\% of the negative target values respectively, chosing random subsets of samples with positive or negative target values. We will refer to the parameter indicating the percentage of samples with positive target values selected as $P$ and to the parameter for the percentage of samples with negative target values selected as $N$. We divided the number of targets covered by each of the 5 mutations equally.

\paragraph{Computing Environment and Solver Configuration.} To describe and solve an ILP we used AMPL 20150516 and CPLEX 12.6.3.  All parameters in CPLEX were left at their default values. We implemented our greedy algorithm in Python 3.6.1. We run our experiments (with the exception of experiments conducted on simulated data) on a MacBook Air with 1.7 GHz Intel Core i7 processor, 8 GB RAM and 500 GB of local storage. In order to make a time comparison with REVEALER we also run the R code provided in~\cite{revealer} on the same machine, using R 3.3.3. All the parameters were left at their given values except for the number of permutations used to calculate their p-value, which we changed in order to compare the running time of the methods excluding the time needed to compute $p$-values.
Experiments on simulated data were conducted on local nodes of a computing cluster. Each node had the following configuration: four 2.27 GHz CPUs, 5.71 GB RAM and 241 GB of storage. 

\section*{Results and discussion}

We tested \algname\ on a number of cancer datasets in order 
to compare its results with state-of-the-art algorithms and to test whether \algname\ allows the 
analysis of large datasets. In particular we used four datasets described in~\cite{revealer} to compare
\algname\ with REVEALER, and then performed a large scale analysis using targets from the Achilles project dataset and alterations from the Cancer Cell Line Encyclopedia (CCLE). 

\subsection*{Comparison with REVEALER}
\label{sec:REVEALER}

We run the greedy algorithm and the ILP from \algname\ on the same four datasets considered by the REVEALER publication~\cite{revealer}. We used the same values of $k$ used in~\cite{revealer}, that is $k=3$ for all the datasets, except from the KRAS dataset where $k=4$ was used. For each dataset we recorded the solution reported by the greedy algorithm, the solution reported by the ILP, the value of the objective functions for such solutions and the running time to obtain such solutions. For ILP solutions, we also performed the permutation test (see Materials and methods) to compute a $p$-value using $1000$ permutations. The results are reported in Table~\ref{Table1}, in which we also show the results from REVEALER (without initial seeds). Fig.~\ref{fig:uncover} shows alteration matrices and the association with the target for the solutions identified by \algname.

\begin{figure*}[h]
    \centering
    \includegraphics[width=\textwidth]{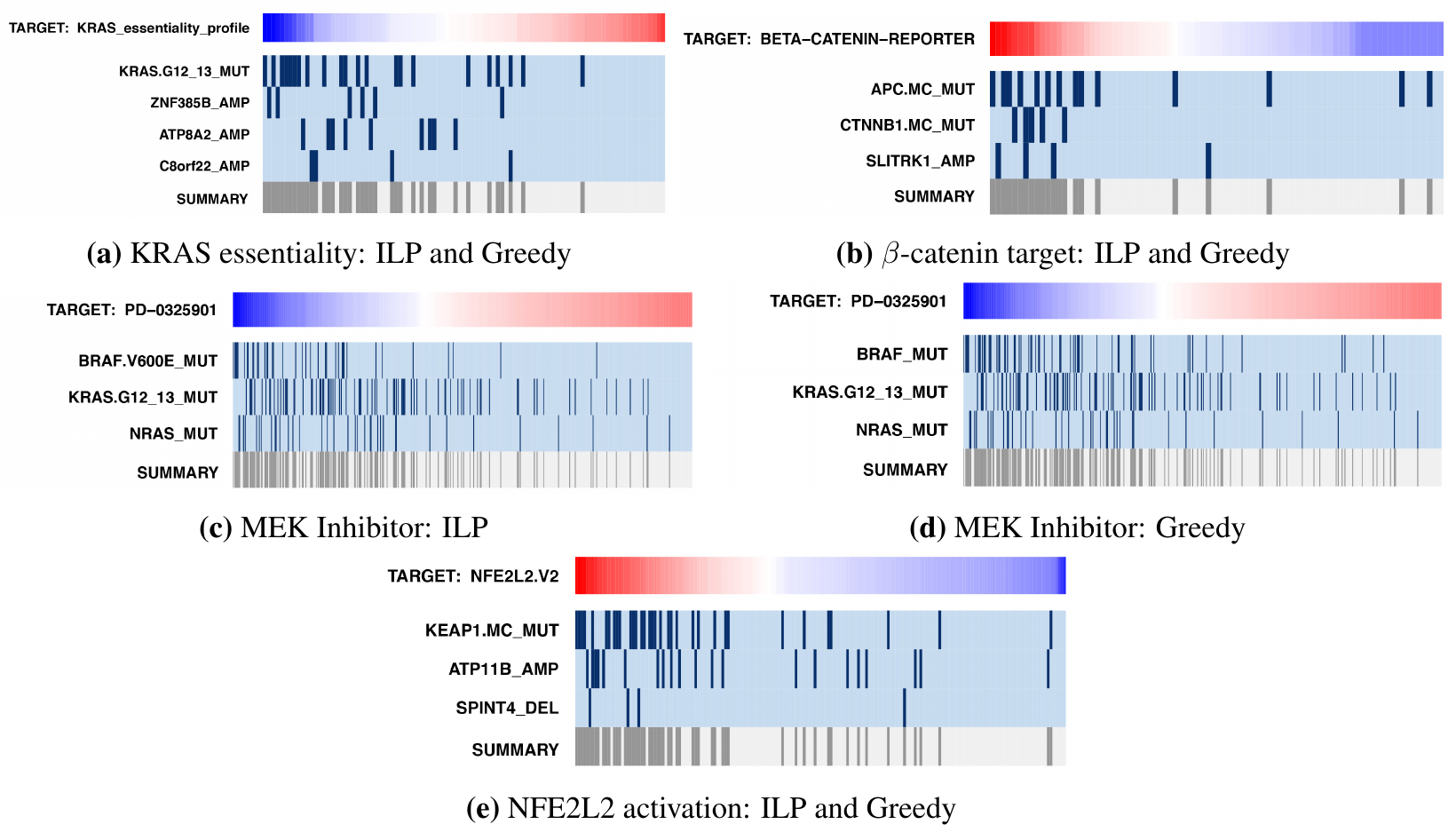}
    \caption{\textbf{Results of \algname  \ on four cancer datasets from~\cite{revealer}.}\\ (a) Solution found by ILP and greedy for KRAS essentiality target. (b) Solution found by ILP and greedy for $\beta$-catenin activation target. (c) Solution found by ILP for MEK inihibitor target. (d) Solution found by greedy for MEK inihibitor target. (e) Solution found by ILP and greedy for NFE2L2 activation target. Each panel shows the value of the target (top row) for various samples (columns), with red being positive and blue being negative values. For each gene in the solution, alterations in each sample are shown in dark blue, while samples not altered are in light blue. The last row shows the alteration profile of the entire solution.
    }
    \label{fig:uncover}
\end{figure*}

We can see that the greedy algorithm identifies the same solution of the ILP based algorithm in three out of four cases, and that the runtime of the ILP and the runtime of greedy algorithm are comparable and very low ($<40$ seconds) in all cases. In contrast, the running time of REVEALER is much higher ($>1000$ seconds in most cases).  (We included all preprocessing in the reported runtimes in table 1 to ensure a fair comparison with REVEALER; not including preprocessing our running times are all under 10 seconds.)  Comparing the alteration matrices of the solutions by \algname\ and the ones of solutions by REVEALER (Supplementary Fig.~\ref{fig:s1}) we note that alterations in solutions by \algname\ tend to have higher mutual exclusivity and to be more concentrated in high weight samples than alterations in solutions by REVEALER. As expected, the value of the objective function we use is much lower for solutions from REVEALER than for solutions from our algorithm.

\begin{table}[!ht]
\begin{adjustwidth}{-2.25in}{0in} % Comment out/remove adjust width environment if table fits in text column.
\centering
\caption{
{\bf Comparison of our algorithms with REVEALER.}}
\begin{tabular}{llllllll}
%		\toprule
		& & \multicolumn{3}{c}{\textbf{ }}  & \\
		\textbf{ }         & \textbf{NFE2L2 activation} & \textbf{MEK-inhibitor} & \textbf{KRAS essentiality} & \textbf{ $\beta$-catenin activation} \\ %\midrule
		\textbf{ILP solution}                   & KEAP1.MC MUT          &  BRAF.V600E MUT          &   KRAS.G12 13 MUT & APC.MC MUT                       \\ 
		\textbf{}           & ATP11B AMP          &  KRAS.G12 13 MUT          &  ZNF385B AMP  & CTNNB1.MC MUT                      \\
		\textbf{}                     & SPINT4 DEL          &  NRAS MUT          &   ATP8A2 AMP       & SLITRK1 AMP                \\ 
		\textbf{}              &           &            & C8orf22 AMP           &                      \\
		\textbf{Objective value}                   & 46.17          &  108.32          &    28.00    & 22.97                    \\
		\textbf{IC score}                    & 0.58          &     0.49       &     0.63      & 0.67                \\
		\textbf{p-value}                    & 0.000999          &   0.000999         &   0.025974     & 0.1068931                   \\
		\textbf{Running time (s)}                  & 14          &   39         &  9     &  9                     \\ %\midrule
		
		\textbf{Greedy solution}                  & KEAP1.MC MUT          &  BRAF MUT          &   KRAS.G12 13 MUT         & APC.MC MUT                \\ 
		\textbf{}                   & ATP11B AMP          &  KRAS.G12 13 MUT          &  ZNF385B AMP     & CTNNB1.MC MUT                     \\
		\textbf{}                    & SPINT4 DEL          &  NRAS MUT          &    ATP8A2 AMP     & SLITRK1 AMP                 \\ 
		\textbf{}                   &           &            &  C8orf22 AMP       &                    \\
		\textbf{Objective value}                   & 46.17          &  104.29          &    28.00        & 22.97                \\
		\textbf{IC score}                    & 0.58          &     0.5       &     0.63       & 0.67               \\
		\textbf{Running time (s)}                  & 15          &   35         &  9          &  8                \\ %\midrule
		
		\textbf{REVEALER  solution}                 & KEAP1.MC MUT          &  BRAF MUT          &   KRAS.G12 13 MUT              & APC.MC MUT            \\ 
		\textbf{}                    & LRP1B DEL          &  KRAS.G12 13 MUT          &  ZNF385B AMP     & CTNNB1.MC MUT                    \\
		\textbf{}                    &  OR4F13P AMP          &  NRAS MUT          &   LINC00340 DEL    & ITGBL1 AMP                    \\ 
		\textbf{}                   &           &            & NUP153 MUT             &               \\
		\textbf{Objective value}                    &   30.35        &  104.29          &    21.86    & 22.12                   \\
		\textbf{IC score}                   & 0.54          &     0.5       &     0.6        & 0.7               \\
		\textbf{Running time (s)}                  &  1615         &  4965          &  1243        &  787                  \\ %\bottomrule
	\end{tabular}
\begin{flushleft} For each of the four targets (NFE2L2 activation, MEK-inhibitor, KRAS essentiality, $\beta$-catenin activation) considered in~\cite{revealer}, the set of alterations of cardinality $k$ reported by our ILP algorithm, by our greedy algorithm, and by REVEALER (without seeds) is reported. $k$ is chosen as in~\cite{revealer}. For each pair (algorithm, target) we also report the (objective) value of our objective function for the solution, the value of the IC score (that is, the objective function used in~\cite{revealer}), and the running time of the algorithm for the target. For solutions found by our ILP we also report the $p$-value computed by permutation test using 1000 permutations.
\end{flushleft}
\label{Table1}
\end{adjustwidth}
\end{table}

We then compared the solutions obtained by our algorihtms with the solutions from REVEALER in terms of the \emph{information coefficient} (IC), that is the target association score used in~\cite{revealer} as a quality of the solution. Surprisingly, in two out of four datasets \algname, which does not consider the IC score, identifies solutions with IC score \emph{higher} (by at least $5\%$) than the solutions reported by REVEALER. For the other two cases, in one dataset the IC score is very similar (0.50 vs 0.49) while in the other case the IC score by REVEALER is higher (0.7 vs 0.67) but the solution reported by REVEALER differs from the solution reported by \algname\ by 1 gene only.
Interestingly, the latter is the only case where the solution from the ILP has a $p$-value $>0.1$ ($p<0.03$ in all other cases), and therefore the solutions (by our methods and by REVEALER) for such dataset may be, at least in part, due to random fluctuations of the data.

In most cases the solutions by \algname\ and by REVEALER are very similar, with cancer relevant genes identified by both methods. For NFE2L2 activation, both methods identify KEAP1, a repressor of NFE2L2 activation~\cite{pmid20534738}. For MEK-inhibitor, both methods find BRAF, KRAS, and NRAS, three well knwon oncogenic activators of the MAPK signaling pathway, which contains MEK as well. For KRAS essentiality, both methods report mutations in KRAS in the solution. For $\beta$-catenin activation, both methods identify CTNNB1 mutations and APC mutations, that is known to be associated to $\beta$-catenin activation~\cite{pmid21859464}.
These results show that \algname\ identifies relevant biological solutions that are better than the ones identified by REVEALER when evaluated using our objective function \emph{and} also when evaluated according to the objective function of REVEALER with a running time that is on average two orders of magnitude smaller than required by REVEALER.

\subsection*{Results on simulated data}

    For each combination we generated 10 simulated dateset as described in Materials and methods. Each dataset contains a \emph{planted} set of 5 alterations associated with the target. We used both the greedy algorithm and the ILP from \algname\  with $k=5$ to attempt to find the 5 correct alteration, and evaluated our algorithms both in terms of fraction of the correct (i.e., planted) solution reported and running time.\\
    
 As shown in Fig.~\ref{fig:runtime}, the greedy algorithm is faster than the ILP for all datasets, and the difference in running time increases as the number $m$ of samples increases, with the runtime of the greedy algorithm being almost two order of magnitude smalle than the runtime of the ILP for $m=1000$ samples. In addition, for a fixed number of samples and alterations, the running time of the greedy algorithm is constant, that is it does not depend on the properties of the planted solution, while the running time of the ILP varies greatly depending on these parameters. For $m=10,000$ samples the running time of the ILP becomes extremely high, so we restricted to consider only two sets of paramters ($p-n=0.95$ and $p-n=0.2$). In this case the ILP took between 44 minutes and 7 hours to complete, while the greedy algorithm terminates in 5 minutes.
 
 In terms of the quality of the solutions found, as expected the ILP outperforms the greedy (Fig.~\ref{fig:solutionquality}) but the difference among the two tends to disappear when the number of samples is higher. In addition, since the ILP finds the optimal solution, we can see that for a limited number of samples we may not reliably identify the planted solution with 200 samples unless the planted solution appears almost only in positive targets and in almost all of them ($p-n=0.95$), while for m=$1000$ we can reliably identify the planted solution using both the ILP and the greedy algorithm even when the association with the target is weaker ($p-n=0.6)$. When $m=10,000$, both the ILP and the greedy algorithm perform well in terms of the quality of the solution: the ILP finds the correct alterations on every experiment and the greedy identifies the whole planted solution in all experiments but one for $p-n=0.2$, for which it still reports a solution containing 4 genes in the planted solution.
 
 These results show that for a large number of samples the greedy algorithm reliably identifies sets of alterations associated with the target, as predicted by our theoretical analysis, and is much faster than the ILP. For smaller sample size the ILP identifies better solutions than the greedy and has a reasonable running time.

 \begin{figure}[t!]
	\centering
		\includegraphics[width=\textwidth]{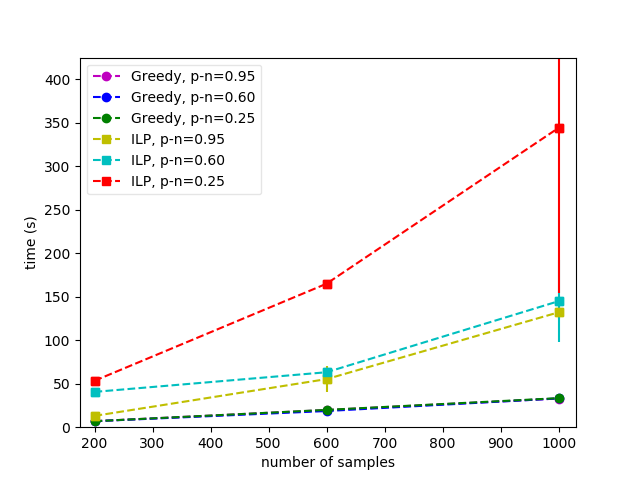}
		\caption{\textbf{Running time of \algname on simulated data.}\\
		The running time (expectation and standard deviation) of the greedy algorithm and of the ILP approach are shown for different number of samples and the difference $p-n$ between the percentage $p$ of samples with positive target and  the percentage $n$ of samples with negative target covered by the the correct solution.}
	\label{fig:runtime}
\end{figure}

\begin{figure}[t!]
	\centering

		\includegraphics[width=\textwidth]{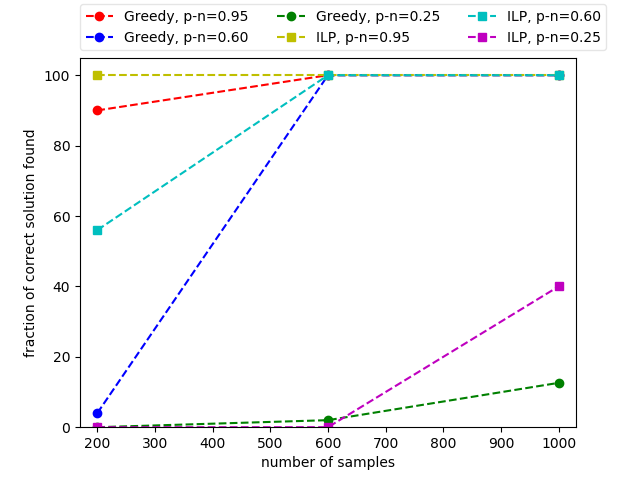}
		\caption{\textbf{Quality of solutions of  \algname on simulated data.}\\
		The fraction of genes in the planted (i.e., correct) solution found by the  greedy algorithm and by the ILP approach are shown for different number of samples and the difference $p-n$ between the percentage $p$ of samples with positive target and  the percentage $n$ of samples with negative target covered by the the correct solution.}
	
	\label{fig:solutionquality}
\end{figure}
 
\subsection*{Analysis of Achilles project data}
\label{sec:achilles}

The efficiency of \algname\ renders the analysis of a large number of targets, such as the ones available through the Achilles project, possible. After preprocessing the dataset included 5690 functional phenotypes as targets, and for each of these the CCLE provides alteration information for 205 samples and 31137 alterations. In total we have therefore run 10380 instances (i.e., 5690 targets screened for positive and for negative associations) looking for both positive and negative association with target values. Since the number of samples (205) is relatively small, we have run only the ILP from \algname\ on the whole Achilles dataset and looked for solutions with $k=3$ genes.
The runtime of \algname\ to find both positive and negative associations, including preprocessing, is 24 hours. Based on the runtime required on the instances reported in~\cite{revealer} (see the Comparison with REVEALER section), running REVEALER on this dataset would have required about 5 months of compute time.

To identify statistically significant associations with targets in the Achilles project dataset we used a nested permutation test. We first run the permutation test with 10 permutations on all instances (i.e., on all targets for both positive association and negative association). We then considered all the instances with the lowest p-value (1/11) and performed a permutation test with 100 permutations only for such instances. We the iterated such procedure once more, selecting all the instances with lowest p-value (1/101) and performing a permutation test with 1000 permutations only for such instances.
For positive association we found 60 solutions with $p$-value $< 0.001$, and for negative association we found 102 solutions with $p$-value $< 0.001$. The solutions with $p$-value $< 0.001$ (with 1000 permutations) are reported in Supplementary Table~\ref{table2} and ~\ref{table3}. See Supplemental Fig.~\ref{fig:s2} for some corresponding alteration matrices.

The genes in the solutions by \algname\ with p-value $1/1001$ are enriched ($p=2\times 10^{-12}$ by Fisher exact test) for well-known cancer genes, as reported in~\cite{vogelstein2013cancer}. We also tested whether genes in solutions by \algname\ (with p-value $1/1001$) are enriched for interactions, by comparing the number of interactions in \texttt{iRefIndex}~\cite{pmid18823568} among genes in such solution with the number of interactions in random sets of genes of the same cardinality. Genes in solutions by \algname\ are significantly enriched in interactions ($p=7\times 10^{-3}$ by permutation test). In addition, the genes in solutions by \algname\ are also enriched in genes in well-known pathways: 12 KEGG pathways~\cite{pmid27899662} have a significant (corrected $p \le 0.05$) overlap with genes in solutions by \algname\ and four of these (endometrial cancer, glioma, hepatocellular carcinoma, EGFR tyrosine kinase inhibitor resistance) are cancer related pathways. In addition, the \emph{targets} (i.e., genes) with solutions of $p$-value $1/1001$ are enriched ($p=10^{-3}$ by permutation test) for interactions in \texttt{iRefIndex} and for well-known cancer genes as reported in~\cite{vogelstein2013cancer}. These results show that \algname\ enables the identification of groups of well known cancer genes with significant associations to important targets in large datasets of functional target profiles. For example, for target (i.e., silenced gene) TSG101, related to cell growth, \algname\ identifies the gene set shown in Figure~\ref{fig:TSG101b} as associated to reduced cell viability. ERBB2 is a well known cancer gene and CDH4 is frequently mutated in several cancer types, and both are associated to cell growth.

\begin{figure}[t!]
	\centering
		\includegraphics[width=0.8\textwidth]{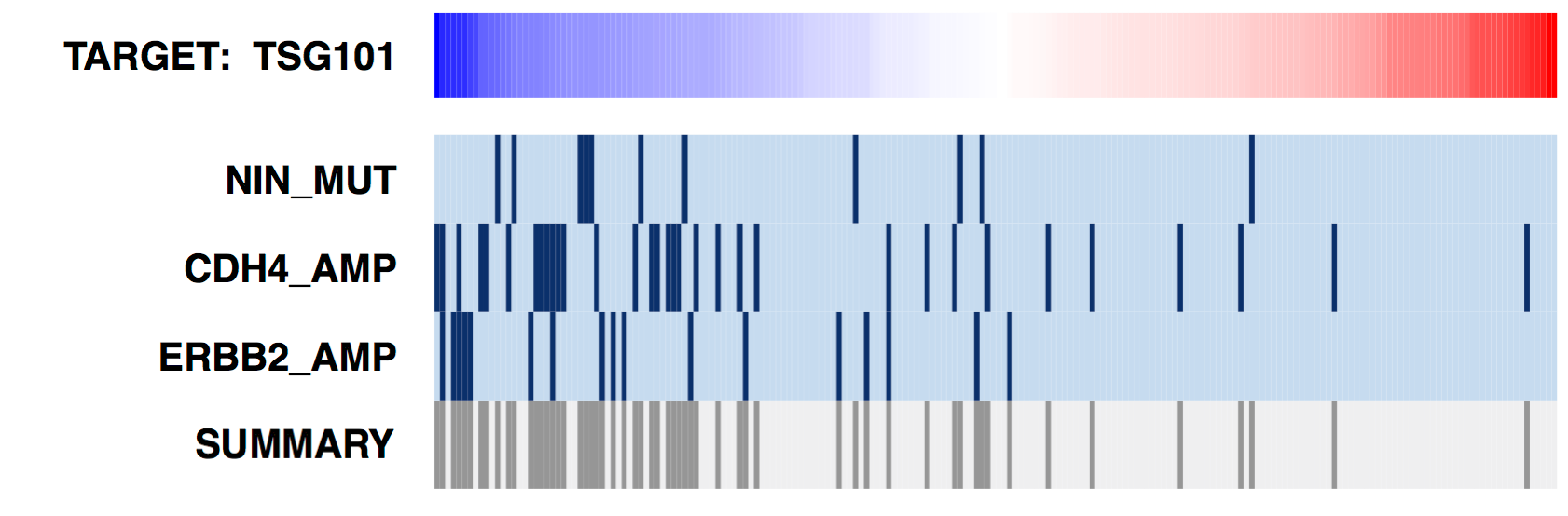}
		\caption{\textbf{Solution by \algname\ for silencing of TSG101  (data from Achilles Project).}\\The alteration matrix of genes in the solution identified by \algname\ as associated to reduced cell viability is reported. The figure shows the value of the target (top row) for various samples (columns), with blue being negative values (i.e., reduced cell viability) and red being positive values. For each gene in the solution, alterations in each sample are shown in dark blue, while samples not altered are in light blue. The last row shows the alteration profile of the entire solution.}
	\label{fig:TSG101b}
\end{figure}

\section*{Conclusion}
\label{sec:conclusion}

In this work we study the problem of identifying sets of mutually exclusive alterations associated with a quantitative target profile. 
We provide a combinatorial formulation for the problem, proving that the corresponding computational problem is NP-hard. We design two efficient algorithms, a greedy algorithm and an ILP-based algorithm, for the identification of sets of mutually exclusive alterations associated with a target profile. We provide a formal analysis for our greedy algorithm, proving that it returns solutions with rigorous guarantees in the worst-case as well under a reasonable genarative model for the data. We implemented our algorithms in our method  \algname, and showed that it finds sets of alterations with a significant association with target profiles in a variety of scenarios. By comparing the results of UNCOVER with the results of REVEALER on four target profiles used in the REVEALER publication~\cite{revealer}, we show that \algname\ identifies better solutions than REVEALER, even when evaluated using REVEALER objective function. Moreover, \algname\ is much faster than REVEALER, allowing the analysis of large datasets such as the dataset from project Achilles, in which \algname\ identifies a number of associations between functional target profiles and gene set alterations.
Our tool \algname\ (as well as REVELEAR) relies on the assumption that the mutual exclusivity among alterations is due to functional complementarity. Another explanation for mutual exclusivity is the fact that each cancer may comprise different subtypes, with different subtypes being characterized by different alterations~\cite{leiserson2013simultaneous}. \algname\ can be used to identify sets of mutually exclusive alterations associated with a specific subtype whenever the subtype information is available, by assigning high weight to samples of the subtype of interest and low weight to samples of the other subtypes.

\section*{Acknowledgements}
  This work is supported, in part, by NSF grant IIS-124758 and by the University of Padova grants SID2017 and Algorithms for Inferential Data Mining, funded by the STARS program. This work was done in part while the authors were visiting the Simons Institute for the Theory of Computing, supported by the Simons Foundation. A short version of this work was accepted to RECOMB 2018.

\nolinenumbers
% Either type in your references using
% \begin{thebibliography}{}
% \bibitem{}
% Text
% \end{thebibliography}
%
% or
%
% Compile your BiBTeX database using our plos2015.bst
% style file and paste the contents of your .bbl file
% here. See http://journals.plos.org/plosone/s/latex for 
% step-by-step instructions.
% 

\newpage
\appendix
{\huge Supplementary Material}

\begin{proposition}
There are instances of the \problemname\ such that $W(\hat{S}) = W(S^*)/k$.
\end{proposition}
\begin{proof}[Proof]
To see that the bound is tight just consider the following example. We want to pick k sets out of n sets $A_1...A_n$. Sets $A_1...A_k$ include 2 elements of respective weight $a\ge 0$ and $b=a/(k-1)$. Subset $A_{k+1}$ includes all the elements of weight $b$ from the  previous $k$ sets and one element with a small weight $\epsilon$. Each of the remaining sets $A_{k+2}...A_n$ include an arbitrary number of elements with overall weight $\le 0$. We choose a penalty of value $a$. Note that one can choose the weights of elements in sets $A_{k+2}...A_n$ in such a way that the average of all positive normalized weights is equal to $a$. Clearly the optimal solution to the \problemname\ problem consists of sets $A_1...A_k$ with an objective value of $k(a+b)$. The greedy algorithm will pick set $A_{k+1}$ at the first iteration and then assign a new weight to its elements equal to $-a$. The updated weight of sets $A_1...A_k$ is now 0 and the algorithm will stop and output $A_{k+1}$ as the solution, giving an approximation ratio of $$\frac {kb+\epsilon} {k(a+b)} = \frac {1} {k} + \frac {\epsilon} {kb}$$
\end{proof}

\begin{proposition}
If $m \in \BOM{k^2 \ln (n/\delta)}$ samples from the generative model above are provided to the greedy algorithm, then the solution of the greedy algorithm is H with probability $\ge \delta$.
\end{proposition}
\begin{proof}[Proof]
We prove that in iteration $i$ of the greedy algorithm, conditioning on the current solution being a set $S$ with $S \subset H$, then the greedy algorithm adds a gene in $H \setminus S$ to the solution with probability $\ge delta/k$, and that the first gene added by the greedy algorithm is $g \in H$. The result then follows by union bound on the $k$ iterations of the greedy algorithm.

Consider the first iteration of the greedy algorithm and consider a gene $g \in G$. Note that if $g \not\in H$ then $\E[W(\{g\})] \le 0$, since $\E[\sum_{j \in A_g} w_j] =0$ because the samples in which $g$ is mutated are taken uniformly at random while $\sum_{j\in A_g} (c_S(j)-1) \ge 0$. If $g\in H$ by the assumptions of the model we have $\E[W(\{g\})] \ge \frac{m}{k c^{'''}}$ for a constant $c'''\ge 1$. Note that $W(\{g\})$ can be written as the sum $\sum_{i=1}^m X_i$ of random variables (r.v.'s) $X_i$ where  $X_i$ is the contribution of sample $i$ to $W(\{g\})$ with $X_i \in [-1,1]$. By the Azuma-Hoeffding inequality~\cite{mitzenmacher2017probability} and union bound (on the $n$ genes) the first gene chosen by the greedy algorithm is not gene $g \in H$ with probability $\le e^{-\frac{2m^2}{4 m k^2 (c^{'''})^2}}$ which is $\le \delta/k$ when $m \in \BOM{k^2 \ln(nk/\delta))}$.

Now assume that in iteration $i$, for the current solution $S \subset H$.  Consider a gene $g \in G \setminus H$, then $\E[W(S \cup \{g\}) - W(S)] \le 0$, since $\E[\sum_{j \in \cup_{s \in S\cup{g}}{A_s}} w_j - \sum_{j \in \cup_{s \in S}{A_s}} w_j] \le 0 $ (by the assumptions of the model $W(S) > 0$ and the fact that alterations in $\{g\}$ are placed uniformly at random among samples) and $\E[\sum_{j \in \cup_{s \in S\cup{g}}}(c_S(j)-1) - \sum_{j \in \cup_{s \in S}}(c_S(j)-1)] \ge 0$ (because for each sample $i$,  the number of alterations of $S \cup \{g\}$ in $i$ is a superset of the number of alterations of $S$ in $i$). Consider now a gene $g \in H \setminus S$: by the assumptions of the model $\E[W(S \cup \{g\}) - W(S)] \le \frac{m}{k c^{'''}}$ for a constant $c^{'''} > 1$. Note that $\E[W(S \cup \{g\}) - W(S)$ can be written as the sum of $\sum_{i=1}^m X_i$ of random variables (r.v.'s) $X_i$ where $X_i$ is the contribution of sample $i$ in the increase in weight from $W(S)$ to $W(S \cup \{g\})$, where $X_i \in [-1,1]$. By the Azuma-Hoeffding inequality and union bound (on the $< n$ genes considered for addition by the greedy algorithm) the gene $g$ added to $S$ by the greedy algorithm in iteration $i$ is not in $H \setminus S$ with probability $\le e^{-\frac{2m^2}{4 m k^2 (c^{'''})^2}}$ which is $\le \delta/k$ when $m \in \BOM{k^2 \ln(nk/\delta))}$.
\end{proof}

\begin{table}[htbp]
	\centering
	\footnotesize
	\caption{Solutions found when running our ILP algorithm for the Achilles dataset, looking for positive correlation with the target. For each target we report the objective function value for the optimal solution, the set of alterations of cardinality 3 and the p-value computed by permutation test using 1000 permutations}
	\begin{tabular}{p{5.915em}rp{7.835em}p{8.585em}p{8.585em}r}
target       & objective & alterations:      &                   &                   \\
VAMP7        & 50.19     & CLYBL\_AMP        & FBXL17\_DEL       & ARHGEF10\_DEL     \\
VHLL         & 43.87     & ZNF705B\_DEL      & FARS2\_DEL        & ZRANB2\_DEL       \\
TRIM13       & 43.75     & AMY2A\_AMP        & SH3KBP1\_DEL      & LOC649352\_DEL    \\
QDPR         & 43.39     & DAB1\_MUT         & SLX1A\_DEL        & RPL23AP53\_DEL    \\
CD302        & 43.36     & FRG2C\_DEL        & LOC650623\_DEL    & CTDP1\_DEL        \\
DDX4         & 43.23     & LYPD8\_AMP        & TRIM28\_AMP       & SLC18A1\_DEL      \\
ZNF439       & 43.21     & PLCXD3\_AMP       & ITGA7\_AMP        & PMP22\_DEL        \\
PTBP3        & 42.58     & TAOK3\_MUT        & MRPS30\_AMP       & ALOX15\_DEL       \\
SRGAP3       & 42.42     & FAM66B\_AMP       & SCAND3\_AMP       & C17orf101\_AMP    \\
C10orf10     & 42.12     & HLA-A\_DEL        & CDH12\_AMP        & PTPRT\_DEL        \\
TNIP1        & 42.10     & LINC00547\_AMP    & XG\_AMP           & LOC729732\_DEL    \\
PAWR         & 42.02     & EPSTI1\_AMP       & GPLD1\_DEL        & DLGAP2\_DEL       \\
PRDM5        & 41.99     & ADAD2\_DEL        & HSD17B10\_DEL     & IMMP2L\_DEL       \\
TXNDC5 & 41.90     & MEP1B\_DEL        & SFRP1\_DEL        & ZNF385D\_DEL      \\
TXNDC5       & 41.90     & MEP1B\_DEL        & SFRP1\_DEL        & ZNF385D\_DEL      \\
ACSL3        & 41.85     & RB1\_MUT          & PLXNA4\_AMP       & PITPNM3\_DEL      \\
ANAPC2       & 41.78     & HERC2P3\_DEL      & LCE1D\_AMP        & MAP1LC3A\_AMP     \\
CNDP2        & 41.76     & EPSTI1\_AMP       & PITPNA\_DEL       & FAM86B2\_DEL      \\
MDM4         & 41.75     & LY86-AS1\_DEL     & THSD7B\_AMP       & LINC00583\_DEL    \\
ADCYAP1R1    & 41.69     & HEATR4\_AMP       & EPB41L4A\_DEL     & DLGAP2\_DEL       \\
NMBR         & 41.64     & LOC100506136\_AMP & AP4S1\_AMP        & PSPC1\_DEL        \\
PPP2R2D      & 41.64     & LZTS1\_DEL        & NF1P2\_DEL        & QKI\_DEL          \\
SLC39A10     & 41.64     & OBSCN\_MUT        & LOC100302640\_AMP & SGSM2\_DEL        \\
GK           & 41.61     & DST\_MUT          & ADCY8\_DEL        & ARHGAP44\_DEL     \\
EIF4E        & 41.36     & GUSBP9\_AMP       & PON2\_AMP         & CAB39L\_DEL       \\
SLC31A1      & 41.26     & FAM66E\_DEL       & TMEM132C\_DEL     & CD83\_DEL         \\
PLS3         & 41.24     & LOC154872\_AMP    & FAM60A\_AMP       & KRT16P2\_DEL      \\
CCT5         & 41.20     & RASA4\_DEL        & GUSBP1\_AMP       & KCNQ5\_DEL        \\
FFAR2        & 41.17     & PTPRT\_DEL        & ERICH1\_DEL       & NBEAP1\_DEL       \\
OR4K17       & 41.15     & THSD7B\_AMP       & C6orf201\_DEL     & FAM86B2\_DEL      \\
TCEB1        & 41.09     & DENND5B\_AMP  & SLC1A3\_AMP       & SLIT3\_DEL        \\
FGD1         & 41.08     & MLL3\_MUT         & MIR19B1\_AMP      & MSR1\_DEL         \\
HNRNPH3      & 40.99     & PIK3CA\_MUT       & BCHE\_AMP         & PRDM2\_DEL        \\
EFHB         & 40.93     & FANCM\_MUT        & WDR7\_DEL         & RIMS2\_DEL        \\
MIF          & 40.90     & FRG2C\_DEL        & LINC00340\_DEL    & PPP3CC\_DEL       \\
PRKY         & 40.85     & ATP11A\_MUT       & TCEB3C\_DEL       & NFIB\_DEL         \\
MGAT4C       & 40.81     & EBF2\_DEL         & ZNF132\_DEL       & KRAS.G12\_13\_MUT \\
RCN2         & 40.54     & MYLK\_MUT         & PIK3C2G\_MUT      & FAM86B1\_DEL      \\
CSE1L        & 40.49     & LY86-AS1\_DEL     & COL1A2\_AMP       & GLIPR2\_DEL       \\
DGKG         & 40.23     & MEP1B\_DEL        & PARD3B\_DEL       & MEX3C\_DEL        \\
MST1R        & 40.20     & SLC6A10P\_DEL     & GOLPH3L\_AMP      & OR52N5\_DEL       \\
TLR4         & 39.91     & C18orf61\_AMP     & PTEN\_DEL         & LOC728875\_DEL    \\
OR4D11       & 39.85     & CREBBP\_MUT       & MIR4796\_AMP      & LOC340357\_DEL    \\
CD47         & 39.71     & STK3\_AMP         & LOC728190\_AMP    & NLRP1\_DEL        \\
RNF183       & 39.71     & GPR112\_MUT       & IFLTD1\_AMP       & LOC340357\_DEL    \\
SULT1A3      & 39.68     & LY86-AS1\_DEL     & CACNA1D\_DEL      & SEC24D\_DEL       \\
SULT1A4      & 39.68     & LY86-AS1\_DEL     & CACNA1D\_DEL      & SEC24D\_DEL       \\
FGG          & 39.66     & H3F3C\_AMP        & MSRA\_DEL         & PSG5\_DEL         \\
NFAT5        & 39.40     & LOC286184\_AMP    & PCTP\_AMP         & MIR4744\_DEL      \\
ANKRD5       & 39.40     & COL14A1\_AMP      & CXADRP2\_DEL      & OPALIN\_DEL       \\
MGP          & 39.37     & TFAP2D\_AMP       & GAL\_AMP          & LOC728323\_DEL    \\
PDE12        & 39.32     & HTR3C\_AMP        & GALNTL2\_DEL      & APC.MC\_MUT       \\
RAB31        & 39.31     & IGLL5\_AMP        & FAM106CP\_DEL     & PACRG\_DEL        \\
ACTR6        & 39.09     & DDAH1\_AMP        & SNORD115-6\_DEL   & LOC340357\_DEL    \\
PGLS         & 38.98     & LOC146880\_AMP    & EYA1\_DEL         & GUCY1A3\_DEL      \\
MAP3K1       & 38.97     & CDH13\_DEL        & ME1\_DEL          & SNTG2\_DEL        \\
SERPINA12    & 38.84     & FLNA\_MUT         & FAM22A\_AMP       & MSR1\_DEL         \\
RWDD2A       & 38.79     & DNAH5\_AMP        & ESR1\_AMP         & COLEC12\_DEL      \\
FAT2         & 38.10     & LDLRAD3\_AMP      & TAC1\_AMP         & LOC440040\_DEL    \\
WSB2         & 36.81     & MLL3\_MUT         & CDH6\_AMP         & CLK2\_AMP
\end{tabular}%
	\label{table2}%
\end{table}% 

\begin{table}[htbp]
	\centering
	\footnotesize
	\caption{Solutions found when running our ILP algorithm for the Achilles dataset, looking for negative correlation with the target. For each target we report the objective function value for the optimal solution, the set of alterations of cardinality 3 and the p-value computed by permutation test using 1000 permutations}
	\begin{tabular}{p{5.915em}rp{7.835em}p{8.585em}p{8.585em}r}

target       & objective & alterations:   &                   &                   \\
KRAS         & 55.22     & KRAS\_MUT      & TUBB8\_AMP        & LRRC37A2\_DEL     \\
COG2         & 50.11     & CRYZ\_DEL      & CCNY\_DEL         & KRAS.G12\_13\_MUT \\
SF3B3        & 46.93     & GPR112\_MUT    & SDHAP1\_AMP       & EIF2C2\_AMP       \\
PIK3CA       & 45.68     & PIK3CA\_MUT    & PRAMEF10\_DEL     & TMEM232\_DEL      \\
BRAF         & 45.65     & BRAF\_MUT      & TPGS2\_DEL        & TRPS1\_DEL        \\
PRPF31       & 45.37     & FBN1\_MUT      & COL22A1\_AMP      & MUC4\_AMP         \\
LOC100131735 & 45.18     & TRPS1\_MUT     & ADCY8\_AMP        & HIPK4\_AMP        \\
SF3A3        & 45.07     & MLL3\_MUT      & KHDRBS3\_AMP      & GPRC5D\_DEL       \\
POLR2A       & 44.99     & EIF2C2\_AMP    & FYTTD1\_AMP       & LOC100506393\_DEL \\
GTF2F1       & 44.92     & KIAA1409\_MUT  & MUC4\_AMP         & FAM135B\_AMP      \\
H1FNT        & 44.19     & CUBN\_MUT      & THSD7B\_AMP       & SLC39A14\_DEL     \\
EIF3I        & 43.88     & MSN\_MUT       & ZFR\_AMP          & MBD2\_DEL         \\
BCLAF1       & 43.43     & DLC1\_MUT      & URB2\_MUT         & FAM83H\_AMP       \\
HAUS1        & 43.33     & ABHD12\_AMP    & HARBI1\_DEL       & MBD2\_DEL         \\
VARS         & 43.13     & MSN\_MUT       & CTNND2\_AMP       & MBD2\_DEL         \\
HGS          & 43.10     & NOSIP\_AMP     & LOC100506990\_DEL & FARS2\_DEL        \\
RPS12        & 43.03     & KHDRBS3\_AMP   & ARL11\_AMP        & TLL2\_DEL         \\
TRIM26       & 42.94     & LRP1B\_MUT     & GSG1L\_DEL        & RUNX1\_DEL        \\
EIF2B3       & 42.87     & ATP7B\_AMP     & RHPN2\_AMP        & MSRA\_DEL         \\
CAND1        & 42.85     & NOTCH1\_MUT    & SLC6A13\_AMP      & KBTBD11\_DEL      \\
TXNL4A       & 42.80     & TXNRD2\_AMP    & PTPN14\_AMP       & CTDP1\_DEL        \\
OC90         & 42.78     & DHRS4L2\_AMP   & SLC6A13\_AMP      & C8orf42\_DEL      \\
UGGT1        & 42.64     & ZNF385B\_DEL   & NLGN1\_DEL        & C18orf26\_DEL     \\
SRRM1        & 42.45     & PIK3CA\_MUT    & SGK2\_AMP         & CELA3A\_DEL       \\
PSMD12       & 42.26     & ZNF286B\_AMP   & GAGE2E\_DEL       & MBD2\_DEL         \\
SNRPF        & 42.25     & PRKCG\_MUT     & COL22A1\_AMP      & MUC4\_AMP         \\
OSR2         & 42.22     & MRPS30\_AMP    & CHEK2P2\_DEL      & TMEM11\_DEL       \\
BUB1B        & 42.18     & FBXO45\_AMP    & EIF2C2\_AMP       & POTEC\_AMP        \\
ADSL         & 41.95     & PARP1\_MUT     & MIR19B1\_AMP      & MIR596\_DEL       \\
ALDH9A1      & 41.88     & LOC728323\_DEL & PRAMEF14\_DEL     & KRAS.G12\_13\_MUT \\
C1QA         & 41.87     & TPPP\_AMP      & DERA\_AMP         & TMEM11\_DEL       \\
RNF40        & 41.86     & KALRN\_MUT     & NF1P2\_DEL        & FAM86B2\_DEL      \\
MEST         & 41.76     & KRAS\_MUT      & EIF3B\_DEL        & DHX38\_DEL        \\
RPAP1        & 41.63     & EIF2C2\_AMP    & SHKBP1\_AMP       & LRRTM4\_DEL       \\
PRPF8        & 41.57     & ITIH5L\_MUT    & PHF20L1\_AMP      & SLC43A2\_DEL      \\
POLR2E       & 41.56     & FBN1\_MUT      & MUC4\_AMP         & FAM135B\_AMP      \\
ABCB7        & 41.50     & ARID1A\_MUT    & RNU6-78\_AMP      & ZNF623\_AMP       \\
POLR2F       & 41.38     & GPR112\_MUT    & EIF2C2\_AMP       & MUC4\_AMP         \\
APLP1        & 41.37     & C3orf33\_AMP   & ERGIC3\_AMP       & CCDC146\_AMP      \\
TSG101       & 41.21     & NIN\_MUT       & CDH4\_AMP         & ERBB2\_AMP        \\
CIAO1        & 41.20     & ZNF705G\_AMP   & MIR3914-1\_AMP    & TBC1D16\_AMP      \\
EIF2B5       & 41.17     & FBN1\_MUT      & SAMD12\_AMP   & PDCD5\_AMP        \\
RANBP2       & 41.09     & KRAS\_MUT      & TRIM49\_AMP       & LOC731275\_AMP    \\
LOC402207    & 41.00     & CTNNB1\_MUT    & LOC284100\_AMP    & LOC649352\_DEL    \\
TOMM40       & 40.94     & EIF2C2\_AMP    & GNAQ\_AMP         & MUC4\_AMP         \\
POLD1        & 40.92     & MLL3\_MUT      & FAM83H\_AMP       & ZNF91\_DEL        \\
LIG4         & 40.91     & TFRC\_AMP      & PLA2G4F\_DEL      & CSGALNACT1\_DEL   \\
MYCBP2       & 40.90     & FLT1\_MUT      & LOC642426\_AMP    & LINC00305\_DEL    \\
SOD1         & 40.82     & ODZ1\_MUT      & ACOT1\_AMP        & LOC729732\_DEL    \\
MAPK4        & 40.77     & NIPBL\_MUT     & FOXP1\_DEL        & PITPNM3\_DEL      \\
TWISTNB      & 40.76     & FGFR2\_MUT     & FBXO32\_AMP       & KRT16P2\_AMP      \\
RBMX         & 40.73     & FBN1\_MUT      & HPYR1\_AMP        & HS1BP3\_DEL       \\
HBG1         & 40.63     & KRAS\_MUT      & FRY\_AMP          & LOC148145\_AMP    \\
EIF2B4       & 40.58     & SGK3\_MUT      & SNORA4\_AMP       & MBD2\_DEL         \\
RBM47        & 40.49     & MIR4470\_AMP   & LOC284100\_AMP    & SETBP1\_DEL       \\
SF3A2        & 40.41     & SETDB1\_MUT    & LRRN3\_DEL        & FAM90A2P\_DEL     \\
NUDT1        & 40.39     & YY1AP1\_AMP    & ANK1\_DEL         & NRN1\_DEL         \\
EIF3F        & 40.30     & LRGUK\_MUT     & EIF2C2\_AMP       & MUC4\_AMP         \\
RPS3         & 40.29     & PIK3CA\_MUT    & LCE1C\_AMP        & MARCH8\_DEL       \\
TUBGCP3      & 40.28     & FCGR1C\_AMP    & URI1\_AMP         & CLDN10\_DEL    

\end{tabular}%
	\label{table3}%
\end{table}%

\begin{table}[htbp]
	\centering
	\footnotesize
	\caption{continues}
	\begin{tabular}{p{5.915em}rp{7.835em}p{8.585em}p{8.585em}r}
target       & objective & alterations:   &                   &                   \\
PRKRIR   & 40.23 & KIAA1549\_MUT  & PCAT1\_AMP        & FAM75A1\_DEL       \\
SF3B1    & 40.21 & FBXO32\_AMP    & ARID3B\_DEL       & FBXW7\_DEL         \\
DDB1     & 40.14 & RGS22\_MUT     & THEG5\_AMP        & WDR7\_DEL          \\
SNRPB2   & 40.08 & ERCC6\_MUT     & GET4\_AMP         & DEFB109P1\_DEL     \\
BRF2     & 40.07 & KRAS\_MUT      & CNTN5\_AMP        & NBEAP1\_DEL        \\
POLR2C   & 39.98 & FLNA\_MUT      & COL22A1\_AMP      & COMMD6\_AMP        \\
POLR2D   & 39.98 & C8orf31\_AMP   & PRKRIR\_AMP       & MBD2\_DEL          \\
PI4KA    & 39.97 & KRAS\_MUT      & MIR624\_AMP       & OR4L1\_DEL         \\
ADRA1B   & 39.96 & ODZ1\_MUT      & PAPL\_AMP         & GTF2E2\_DEL        \\
RFX2     & 39.93 & OBSCN\_MUT     & TBL1XR1\_AMP      & FAM18B2\_DEL \\
SFPQ     & 39.93 & USP25\_AMP     & TRIM49\_AMP       & PYCRL\_AMP         \\
PSMD13   & 39.92 & NIN\_MUT       & VAPB\_AMP         & KIAA0825\_DEL      \\
KRR1     & 39.87 & EVPLL\_AMP     & TNXB\_AMP         & MBD2\_DEL          \\
LGI1     & 39.86 & SIK1\_AMP      & CDC73\_AMP        & FAM221A\_DEL       \\
COPZ1    & 39.83 & EIF2C2\_AMP    & MUC4\_AMP         & ZNF91\_DEL         \\
PDXK     & 39.76 & COX10-AS1\_DEL & SNTG2\_DEL        & OR4K5\_DEL         \\
COPS2    & 39.75 & TRMT12\_AMP    & BMPR2\_AMP        & STS\_DEL           \\
LSM3     & 39.72 & ZC3H3\_AMP     & MUC4\_AMP         & EXOC4\_DEL         \\
APOBEC3G & 39.48 & SNAR-D\_AMP    & TPPP\_AMP         & APC.MC\_MUT        \\
TCOF1    & 39.39 & HCN1\_AMP      & PTPRT\_DEL        & PRR5\_DEL  \\
RAB19    & 39.38 & PIK3CA\_MUT    & FLJ31813\_DEL     & MTHFD1\_DEL        \\
ZNF433   & 39.38 & PCDH15\_MUT    & DPF2\_AMP         & ATP2A3\_DEL        \\
OXSM     & 39.26 & TSHZ3\_AMP     & C17orf101\_AMP    & ELAC1\_DEL         \\
AP3M1    & 39.16 & SMARCA4\_MUT   & CALM1\_AMP        & MBD2\_DEL          \\
PLRG1    & 39.14 & FBN1\_MUT      & MUC4\_AMP         & TSTA3\_AMP         \\
SRP9     & 39.09 & EDN3\_AMP      & RSBN1L\_AMP       & NF2\_DEL           \\
RPS27A   & 39.07 & NOS2\_MUT      & TPPP\_AMP         & C17orf51\_DEL      \\
USPL1    & 39.01 & LOC643401\_AMP & TOMM20\_AMP       & ARL8B\_DEL         \\
ARHGEF3  & 38.99 & ZNF423\_MUT    & ST6GAL1\_DEL      & FBXO31\_DEL        \\
SNW1     & 38.69 & MLLT3\_MUT     & WRN\_MUT          & KIFC2\_AMP         \\
OTUD7A   & 38.62 & DST\_MUT       & PTPN21\_AMP       & MIR497HG\_DEL      \\
FUT6     & 38.43 & ACLY\_AMP      & FAM190A\_DEL      & ADAMTSL3\_DEL      \\
SLC25A20 & 38.41 & DOCK3\_MUT     & ZNF536\_AMP       & FHOD3\_DEL         \\
EFTUD2   & 38.32 & FBN1\_MUT      & MUC4\_AMP         & FAM135B\_AMP       \\
IFT27    & 38.31 & SVIL\_AMP      & C18orf26\_DEL     & HTT\_DEL           \\
PRKACG   & 37.98 & UBE3B\_MUT     & LOC100287314\_AMP & ZNF705B\_DEL       \\
RUFY1    & 37.86 & AHCY\_AMP      & PRICKLE4\_AMP     & NETO1\_DEL         \\
DHX15    & 37.83 & KIAA2022\_MUT  & SNAR-I\_AMP       & SERPINB11\_DEL     \\
NCBP2    & 37.77 & GPR112\_MUT    & FCGR3B\_AMP       & FAM135B\_AMP       \\
ESPL1    & 36.84 & GPC6\_AMP      & CD226\_DEL        & KALRN\_DEL         \\
RPL31    & 36.09 & PIK3CA\_MUT    & LOC643401\_AMP    & FLJ23152\_DEL      \\
NHP2L1   & 35.98 & CEP72\_AMP     & CT45A3\_DEL       & COLEC12\_DEL    

\end{tabular}%
	\label{tab:addlabel}%
\end{table}

\begin{figure*}[h]
    \centering
    \includegraphics[width=\textwidth]{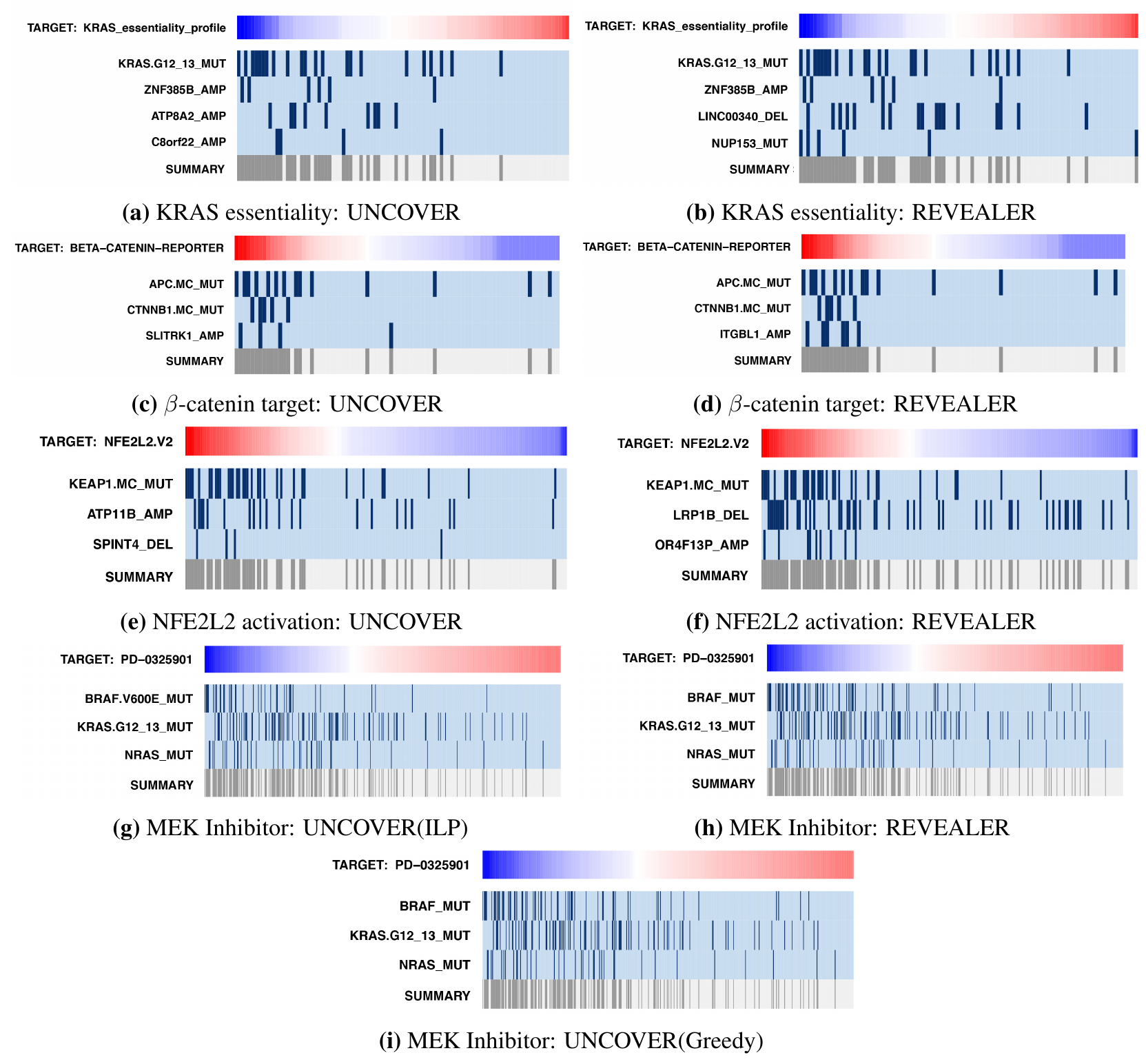}
    \caption{\textbf{Results of \algname\ and REVEALER on four cancer datasets from~\cite{revealer}.}\\ (a) Solution found by ILP and greedy for KRAS essentiality target. (b) Solution found by ILP and greedy for $\beta$-catenin activation target. (c) Solution found by ILP for MEK inihibitor target. (d) Solution found by greedy for MEK inihibitor target. (e) Solution found by ILP and greedy for NFE2L2 activation target. Each panel shows the value of the target (top row) for various samples (columns), with red being positive and blue being negative values. For each gene in the solution, alterations in each sample are shown in dark blue, while samples not altered are in light blue. The last row shows the alteration profile of the entire solution.
    }
    \label{fig:s1}
\end{figure*}

\begin{figure*}[h]
    \centering
    \includegraphics[width=0.7\textwidth]{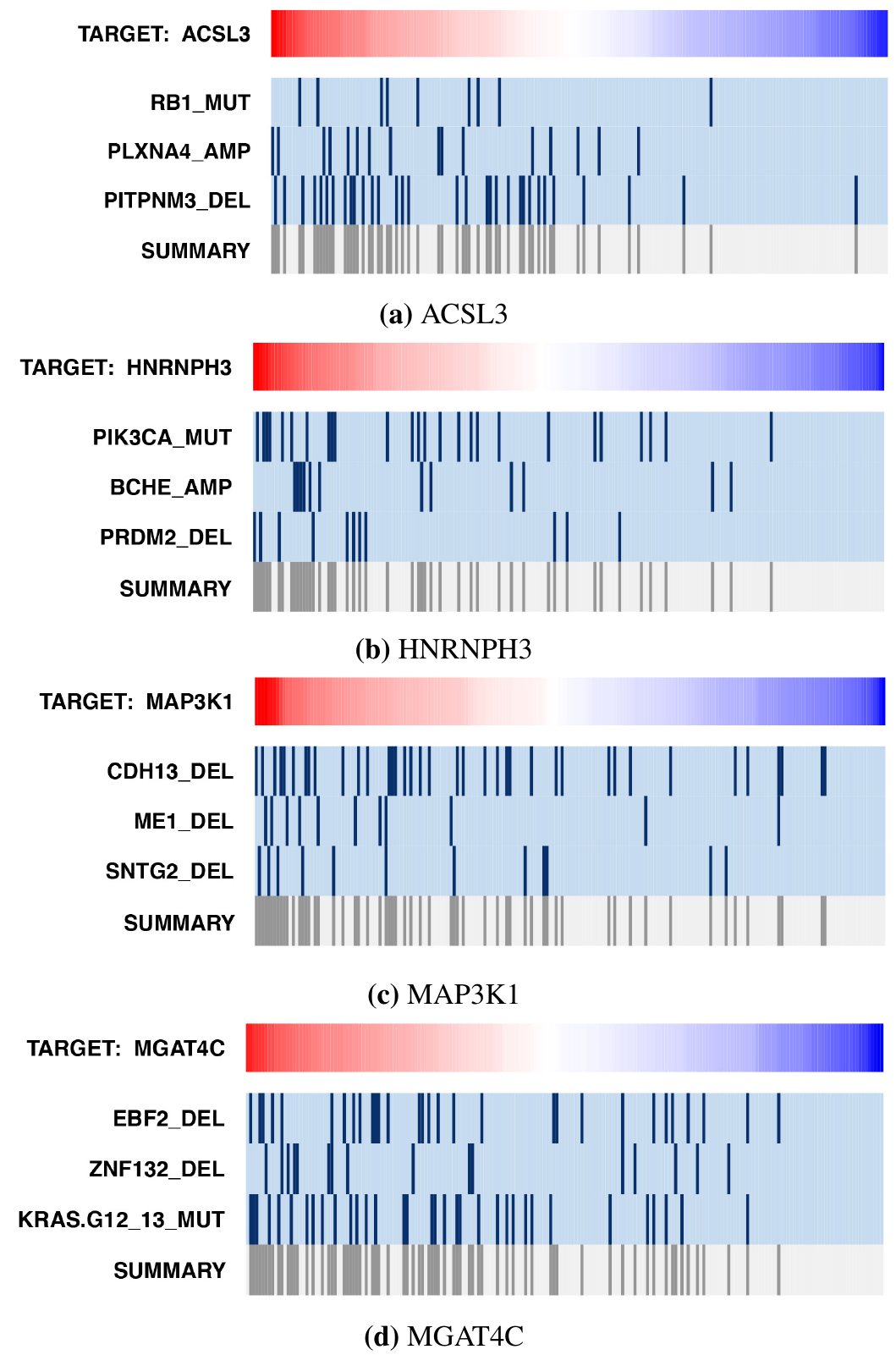}
    \caption{\textbf{Alteration matrices for results of \algname\ the Achilles Project data.}
    }
    \label{fig:s2}
\end{figure*}

\end{document}